\documentclass[12pt]{article}

\usepackage{graphicx}

\newcommand{\blackslug}{\hbox{\hskip 1pt
        \vrule width 4pt height 8pt depth 1.5pt\hskip 1pt}}
\newcommand{\myQED}{\hfill \blackslug}

\newenvironment{proof}
    {\pagebreak[1]{\narrower\noindent {\bf Proof:\nopagebreak}}}%
    {\myQED}

    {\myQED}

\newtheorem{lemma}{Lemma}

\begin{document}

\begin{center}
{\large \bf
Exploring Viable Algorithmic Options For \\
Automatically Creating and Reconfiguring \\ 
Component-based Software Systems: \\
A Computational Complexity Approach \\
(Full Version)}

\vspace*{0.2in}

Todd Wareham \\
Department of Computer Science \\
Memorial University of Newfoundland \\
St.\ John's, NL Canada \\
(Email: {\tt harold@mun.ca}) \\
\vspace*{0.1in}

Marieke Sweers \\
Department of Artificial Intelligence and Psychology \\
Radboud Universiteit Nijmegen \\
Nijmegen, The Netherlands \\
(Email: {\tt marieke.sweers@gmail.com}) \\
\vspace*{0.1in}

\today
\end{center}

\begin{quote}
{\bf Abstract}:
Component-Based Development (CBD) is a popular approach to mitigating the costs 
of creating software systems. However, it is not clear to what extent the core 
component selection and adaptation activities of CBD can be implemented to operate 
automatically in an efficient and reliable manner or in what situations 
(if any) CBD is preferable to other approaches to software development. In this paper, 
we use computational complexity analysis to determine and compare the computational 
characteristics of fully automatic component-based software system creation and reconfiguration by 
{\em de novo} design, component selection, and component selection with adaptation. Our
results show that none of these approaches can be implemented to operate both 
efficiently and reliably in a fully automatic manner either in general or relative to a number of
restrictions on software systems, system requirements, components, and component 
adaptation. We also give restrictions under which all of these approaches can be 
implemented to operate both efficiently and reliably. As such, this paper
illustrates how different types of computational complexity analysis (in
particular, parameterized complexity analysis) can be used to systematically
explore the algorithmic options for implementing automatic activities in software
engineering.
\end{quote}

\section{Introduction}

\label{SectIntro}

Since its proposal by McIlroy in 1968 \cite{McI68}, Component-Based 
Development (CBD) 
has been a popular approach to mitigating the monetary and time costs of creating 
and maintaining software systems \cite{FK05,GMS02,HC01,SH04,VC+16}. In CBD, 
previously-developed software modules called components are connected together (possibly
with some adaptation of the component code and architecture) to generate new 
software systems from given requirements. In this process, components, software
systems, and the manners in which components can be used to create these system are specified 
using standardized models of components, component-based architectures, and component
composition. Over the last 20 years, a number of 
technologies have been implemented to assist human programmers in creating 
component-based systems, e.g., CORBA, CCM, EJB, COM+, and much effort has been put into
automating the key CBD activities of component selection and adaptation 
\cite{GMS02,Kel08,VC+16}.

There are two ongoing issues of great importance in the CBD community: (1) the desire 
to (in the case of ubiquitous and autonomic computing systems \cite{BB02,KC03}, fully) 
automate core CBD activities such as component selection and adaptation 
\cite{FK05,Kel08} and (2) the circumstances (if any) in which the cost of developing 
software systems by selecting and adapting components is less than that of designing 
systems from scratch, i.e., {\em de novo} \cite{GT05,MMM95}. The second issue is of 
particular importance because ``[e]conomic considerations, and cost/benefit analyses in
general, must be at the center of any discussion of software reuse'' 
\cite[page 557]{MMM95}. In addition to these general issues, there are also
computational conjectures about specific CBD activities based on empirical observations 
that need to verified, e.g., whether increases in the computational effort required
to adapt components is due to increases in the size of the code being modified and/or 
the number of code modifications increases \cite[Page 537]{MMM95}.

To address all of these issues and questions, it would be most useful to know the range
of possible methods for implementing efficient and reliable automated CBD as well as 
the circumstances (if any) in which these methods outperform other approaches to 
software development. This is best done using computational complexity analysis 
\cite{GJ79,DF13}. Such analyses determine whether or not there is an 
efficient algorithm for a given problem, i.e., whether that problem is tractable 
or intractable. So-called classical complexity analysis \cite{GJ79} establishes 
whether a problem can be solved efficiently in general, i.e., for all inputs. 
If this is not the case, parameterized complexity analysis \cite{DF13} establishes 
relative to which sets of input and output restrictions the problem can and cannot be solved efficiently.
Parameterized analyses thus allow the systematic exploration of viable algorithmic options
for solving an intractable problem \cite{War99}, cf.\ the piecemeal 
situation-by-situation manner in which algorithms for a problem are
typically derived.
In order to have the greatest possible applicability, both of these types of analyses are
typically performed relative to simplified versions of problems that are special cases of
the actual problems of interest. This is because, as an efficient algorithm for the 
general case can also efficiently solve special cases of that general case, intractability
results for special cases also imply intractability for all more general cases (see 
Section \ref{SectGenRes} for details). 

In this paper, we address the issues above by completing and extending initial
work reported (without almost all proofs of results) in \cite{WS15} 
(subsequently reprinted as \cite{WS16}).
We first use computational complexity analysis to show that neither 
creating nor reconfiguring software systems either by {\em de novo} design, component 
selection, or component selection with adaptation are tractable in 
general by either deterministic or probabilistic algorithms. We then use parameterized 
complexity analysis to give restrictions under which all of these activities are tractable and prove that surprisingly few restrictions on 
either the structure of software systems, system requirements, and components or the 
allowable degree of adaptation render these 
activities tractable. Taken together, these sets of results give the first objective
comparison of the effort involved in both creating and reconfiguring software systems
by {\em de novo} design, component selection, and component selection with adaptation. 
Our intent is to derive results that have the greatest possible applicability
to component-based software system development.
Hence, by the logic in the previous paragraph, our results are derived 
relative to simple models of 
software systems, system requirements, components, and adaptation (namely, purely
reactive two-level systems that satisfy specified functional requirements and whose 
whitebox procedure-components are adapted by changes to their code) which are special cases
of a broad range of more realistic models.

\subsection{Previous Work}

\label{SectPrevWork}

A number of problems in software engineering are known to be undecidable, i.e.,
do not have any algorithm that works correctly for all inputs. This is a
consequence of Rice's Theorem \cite[Section 9.3.3]{HMU01}.\footnote{
We thank an anonymous reviewer of an earlier version of this paper for
pointing this out.
}
One such result
is that there is no algorithm which, given a software system and a set of
requirements, can correctly decide whether or not that system satisfies
those requirements. That being said, it is also known that Rice's Theorem
only applies to the most general statements of such problems 
\cite[Section 9.3.4]{HMU01}, and that special cases of these problems (such as
those examined in this paper) may still have correct and even efficient
algorithms.

Computational complexity analyses have been done previously for component selection
\cite{PO99,PWM03} and component selection with adaptation \cite{BBR05}, with 
\cite{BBR05,PWM03} having the additional requirement that that number of components in
the resulting software system be minimized. Given the intractability of all of these
problems, subsequent work has focused on efficient approximation algorithms for
component selection. Though it has been shown that efficient algorithms that produce 
software systems whose number of components is within a constant multiplicative factor 
of optimal are not possible in general \cite{NH08}, efficient approximation algorithms 
are known for a number of special cases \cite{FBR04,HM+07,NH08}. 

All of these analyses use atomic blackbox formalizations of system requirements, 
components, and (in the case of \cite{BBR05}) component adaptation. While this is 
consistent with component selection as typically employed in both academia and industry 
\cite{AH+11,LB+08,MRE07}, none of these formalizations include specifications of 
the internal structure of software systems, system requirements, and components (and 
hence component adaptation as well) that are detailed enough to allow
investigation of restrictions on these aspects that could make component selection or 
component selection with adaptation tractable. Moreover, all work to date has
focused on component selection and adaptation in the context of system creation
and has not considered subsequent modification to a system over that system's lifetime.

Various metrics measuring the relative costs and hence benefits of developing software
systems by {\em de novo} design, component selection, and component selection with
adaption have been proposed (see \cite[Section 3]{MMM95} and references). 
Problems associated with subjectively estimating certain cost-variables within these 
metrics, e.g., adapting components and developing reusable components, led Mili et al 
in 1995 to conclude that ``\ldots we will guesstimate \ldots [the] relative 
effectiveness [of various software engineering methods] but we will not, {\em and 
cannot}, go any further \cite[page 540]{MMM95} (emphasis added). The situation has
not improved significantly in the last 20 years, as both Frakes and Kang in 2005
\cite[Section 4]{FK05} and Vale et al in 2016 \cite[Page 143]{VC+16} concluded
that the lack of metrics related to CBD was still a notable gap in CBD research.
This has been mitigated by the several hundred studies and experiments 
done over the last several decades to validate proposed CBD methods and evaluate their 
effectiveness in various situations \cite[Section 3.5]{VC+16}. However, to our 
knowledge, there has been no objective mathematically-based assessment to date of the 
relative computational effort required to create let alone reconfigure software systems 
by {\em de novo} design, component selection, and component selection with adaptation.

\subsection{Organization of Paper}

Our paper is organized as follows. In Section \ref{SectForm}, we present our models
of software system requirements, component-based architectures, components, 
component composition, and component adaptation and formalize the problems
of {\em de novo} and component-based software system creation and reconfiguration relative
to these models.  Section \ref{SectCIRes} demonstrates the intractability of all of 
these problems in general and Section \ref{SectPrmRes} establishes several basic sets 
of restrictions under which our problems are and are not tractable. In order to focus 
in the main text on the implications of our results for CBD, proofs of all results are 
given in an appendix. The generality and implications of our results are 
discussed in Sections \ref{SectGenRes} and \ref{SectDisc}, respectively. Finally, our 
conclusions and directions
for future work are given in Section \ref{SectConc}.

\section{Formalizing Component-based Software System \\ \hspace*{0.2in} Creation and Reconfiguration}

\label{SectForm}

\subsection{Issues addressed in our formalizations}

\label{SectFormIssues}

One of our goals in this paper is to assess whether or not component-based development 
has advantages over other types of software development relative to different activities
in the software life cycle. This assessment can be stated most simply along two 
dimensions:

\begin{enumerate}
\item{
\textit{Software design mode}: creating software using selected components versus 
creating software by adapting selected components versus creating software {\em de novo}
that is organized in a specified component-like fashion.
}
\item{
\textit{Software lifecycle activity}: Creating the initial version of
a software system relative to a set of requirements versus reconfiguring
an existing system to accommodate one or more changes to the requirements.\footnote{
Note that these creation and reconfiguration activities are invoked in
implementing the self-configuration and self-healing properties of autonomic
systems \cite{KC03} as well as the automatic synthesis of appropriate applications for
selected tasks in ubiquitous computing \cite{BB02}.
}
}
\end{enumerate}

\noindent
The possibilities implicit in these dimensions result in the following
six informal computational problems:

\vspace*{0.1in}

\noindent
{\sc Software Creation} \\
{\em Input}: Software requirements $R$, software-structure specification $X$. \\
{\em Output}: A software system $S$ whose structure is consistent with $X$
               and whose operation satisfies $R$.

\vspace*{0.1in}

\noindent
{\sc Software Creation from Components} \\
{\em Input}: Software requirements $R$, a set of software-component libraries
              $\cal{L}$ $= \{L_1, L_2,$ \linebreak $\ldots, L_{|C|}\}$. \\
{\em Output}: A software system $S$ whose structure consists of 
               components drawn from the libraries in $\cal{L}$ and 
               whose operation satisfies $R$.

\vspace*{0.1in}

\noindent
{\sc Software Creation from Adapted Components} \\
{\em Input}: Software requirements $R$, a set of software-component libraries
              $\cal{L}$ $= \{L_1, L_2,$ \linebreak $\ldots, L_{|C|}\}$. \\
{\em Output}: A software system $S$ whose structure consists of 
               components derived  from components drawn from the libraries in 
               $\cal{L}$ and whose operation satisfies $R$.

\vspace*{0.1in}

\noindent
{\sc Software Reconfiguration} \\
{\em Input}: A software system $S$ whose structure is consistent with 
              specification $X$ and whose operation satisfies requirements $R$,
              new software requirements $R_{new}$. \\
{\em Output}: A software system $S'$ derived from $S$ whose structure is 
               consistent with $X$ and whose operation satisfies $R \cup R_{new}$.

\vspace*{0.1in}

\noindent
{\sc Software Reconfiguration from Components} \\
{\em Input}: A software system $S$ whose structure consists of 
              components from a set of software-component libraries
              $\cal{L}$ $= \{L_1, L_2, \ldots, L_{|C|}\}$ and whose operation
              satisfies requirements $R$, new software requirements $R_{new}$. \\
{\em Output}: A software system $S'$ derived from $S$ whose structure consists
               of components drawn from the libraries in $\cal{L}$ 
               and whose operation satisfies $R \cup R_{new}$.

\vspace*{0.1in}

\noindent
{\sc Software Reconfiguration from Adapted Components} \\
{\em Input}: A software system $S$ whose structure consists of 
              components from a set of software-component libraries
              $\cal{L}$ $= \{L_1, L_2, \ldots, L_{|C|}\}$ and whose operation
              satisfies requirements $R$, new software requirements $R_{new}$. \\
{\em Output}: A software system $S'$ derived from $S$ whose structure consists
               of components derived from components drawn from the 
               libraries in $\cal{L}$ and whose operation satisfies $R \cup R_{new}$.

\vspace*{0.1in}

\noindent
We want to be able to assess the effects of as many characteristics of the software 
system specifications, requirements, and structure as possible on the computational 
difficulty of the creation and reconfiguration processes. This requires that explicit 
and detailed representations of these entities be included in our problems. This is very
different from the more abstract atomic blackbox perspective adopted by previous 
computational complexity analyses of CBD when defining their problems (see Section 
\ref{SectPrevWork}). However, as we will see later in this paper, this will allow our 
complexity analyses (in particular, our parameterized complexity analyses) to build on 
and extend the results derived in those previous complexity analyses.

\subsection{Formalizations of problem entities}

\label{SectFormEnt}

To assess the general computational difficulty of and (if necessary) explore algorithmic
options under restrictions for efficiently solving the problems sketched in Section 
\ref{SectFormIssues}, we need formalizations of all of the entities comprising these
problems. There is a vast existing literature on various models of software system 
requirements, component-based architectures, components, component composition, 
component libraries, and component adaptation
\cite{GMS02,HC01,Kel08}. We want to avoid introducing spurious computational difficulty
due to powerful formalisms, as this would complicate the interpretation of our results;
we also want our results to have the greatest possible applicability in the sense 
described in Sections \ref{SectIntro} and \ref{SectGenRes}. Hence, we shall choose the following basic 
formalizations:\footnote{
These formalizations are based on those developed in \cite{OS+15} and \cite{Swe15} (see also
\cite{OB+18}) to analyze the computational complexity of the adaptive toolbox theory of
human decision-making \cite{GT99}. 
}

\begin{itemize}
\item{
\textit{Software system requirements}: The requirements 
will be a set $R = \{r_1, r_2, \ldots r_{|R|}\}$ of situation-action pairs 
where each pair $r_j = (s_j,a_j)$ consists of a situation $s_j$ defined by a 
particular set of truth-values $s_j = \langle v(i_1), v(i_2), \ldots$
$v(i_{|I|})\rangle$, $v(i_k) \in \{True, False\}$, relative to each of the Boolean 
variables $i_k$ in set $I = \{i_1, i_2, \ldots, i_{|I|}\}$ and an action $a_j$ from
set $A = \{a_1, a_2, \ldots a_{|A|}\}$. As such, these are
functional requirements describing wanted system behaviors.
}
\item{
\textit{Component-based architecture model}: We will here consider a two-level software system architecture 
consisting of a top-level multiple 
\texttt{IF-THEN-ELSE} statement block (a {\bf selector}) whose branches are in turn 
lower-level \linebreak \texttt{IF-THEN-ELSE} statement blocks ({\bf procedures}) whose branches 
execute actions from $A$. Each selector and procedure \texttt{IF-THEN}
condition is a Boolean formula that is either a variable from $I$ 
e.g., $i_j$, or its negation, e.g., $\neg i_j$. 
Each selector and procedure with one or more situation-variable conditions
terminates in a final \texttt{ELSE} statement. 
In addition, the selector may consist of single special
statement \texttt{IF * THEN} which evaluates to $True$ in all cases
(the default selector) and a procedure may consist of a single executed
action; such selectors and procedures have no associated conditions.
}
\item{
\textit{Component, component composition, and component library models}: Two-level
software systems of the form above consist of two
types of components, selectors and procedures. 
The code for each component is available, i.e., they are
whitebox components, and the behavior of a component is specified
by its code. 
Components are composed into software systems using only simple procedure calls that 
do not either have input parameters or return values at procedure termination,
These components are
stored in libraries $L_{sel}$ and $L_{prc}$, respectively.
}
\item{
\textit{Software structure specification}: The basic characteristics
of two-level software systems of the form above are
the situation-variable and action sets on which they are based ($I$ and $A$),
the maximum number of conditions allowed in the selector
($|sel|$) and the maximum number of conditions allowed in
any procedure ($|prc|$).
Note that the maximum possible size of such a system (described
in terms of the number of lines of code in the system) is captured by
the expression $(|sel| + 1)(|prc| + 2)$.
}
\item{
\textit{Software system adaptation}: 
We will here consider two types of adaptation: (1) changes to the system code: 
changing the condition in or action executed by any \texttt{IF-THEN-ELSE} statement 
(component adaptation) and (2) changes to the system component-set: using a 
different selector-component from $L_{sel}$ or changing any used procedure-component 
to another from $L_{prc}$ (system adaptation). Note in the case of a selector-component
change of type (2), both the original and new selector must have the same number of 
\texttt{IF-THEN} statements and the procedures called by the original selector are 
mapped to the corresponding positions in the new selector. As such, these adaptations 
are invasive and at the code level (though, in a limited sense, adaptations of type (2)
can be seen as operating at the architecture level).
}
\end{itemize}

\begin{figure}[p]
\small
\centering
\begin{tabular}{| c || c | c | c | c | c || c |}
\hline
req.\ & $i_1$ & $i_2$ & $i_3$ & $i_4$ & $i_5$ & action \\
\hline\hline
$r_1$ & T & T & T & T & T & $a_2$ \\
\hline
$r_2$ & T & F & F & F & T & $a_1$ \\
\hline
$r_3$ & F & F & F & F & F & $a_2$ \\
\hline
$r_4$ & F & F & F & F & F & $a_2$ \\
\hline
$r_5$ & T & T & T & F & T & $a_3$ \\
\hline
\end{tabular} \\

{\tt
\begin{tabular}{l l}
 & \\
\multicolumn{2}{c}{(a)} \\
 & \\
procedure p1:                   & procedure p2 \\
~ if $i_4$ then $a_2$           & ~ if $\neg i_2$ then $a_1$ \\
~ elsif $\neg i_3$ then $a_1$ & ~ elsif $\neg i_4$ then $a_2$ \\
~ elsif $i_5$ then $a_3$      & ~ els$a_3$ \\
~ else $a_1$                    & \\
 & \\
procedure p3:                   & procedure p4: \\
~ if $i_4$ then $a_2$           & ~ $a_2$ \\
~ else $a_2$                    & \\
 & \\
\multicolumn{2}{c}{(b)} \\
 & \\
selector s1:                    & selector s2: \\
~ if $i_1$ then ???             & ~ if * then ??? \\
~ elsif $i_5$ then ???        & \\
~ else ???                      & \\
 & \\
\multicolumn{2}{c}{(c)} \\
 & \\
system S1:                    & system S2: \\
~ if $i_1$ then call p1         & ~ if * then call p2 \\
~ elsif $i_5$ then call p3    & \\
~ else call p4                  & \\
 & \\
\multicolumn{2}{c}{(d)} \\
 & \\
\end{tabular}
}
%\end{center}
\caption{ Example Requirements, Procedures, Selectors, and Software Systems. 
           (a) Software requirements $R = \{r_1, r_2, r_3, r_4, r_5\}$ defined on
           situation-variables $I = \{i_1, i_2, i_3, i_4, i_5\}$ and action-set 
           $A = \{a_1, a_2, a_3\}$. (b) Four procedures. (c) Two selectors as they would
           appear in a selector-component library with blank procedure calls, where $s2$
           is the default selector. (d) Two software systems created by instantiating 
           the procedure-calls in the selectors from part (c) with procedures from part (b).
}
\label{FigCSS}
\end{figure}

\noindent
Examples of the first four of these entities are given in Fig. \ref{FigCSS}.
The above allows us to formalize various actions and properties in the
problems given in Section \ref{SectFormIssues} as follows:

\begin{itemize}
\item A software system $S$ is consistent with a software structure-specification
       $X = \langle I, A, |sel|, |prc|\rangle$ if it has the two-level structure 
       described above where \texttt{IF-THEN} conditions are members of $I$, all 
       procedure-executed actions are drawn from $A$, the number of conditions in the 
       selector is at most $|sel|$, and the number of conditions in each procedure is 
       at most $|prc|$. Note that the default selector has $|sel| = 0$ (e.g., selector 
       {\tt s4} in Fig. \ref{FigCSS}(c)) and a procedure consisting of a single action
       has $|prc| = 0$ (e.g., procedure {\tt p4} in Fig. \ref{FigCSS}(b)). For 
       example, both software-systems in Fig. \ref{FigCSS}(d) are consistent with $X =
       (I = \{i_1, i_2, i_3, i_4, i_5\}, A = \{a_1, a_2, a_3\}, |sel| = 2, |prc| = 3)$. 
\item A software system $S$ is constructed from components drawn from 
       $\cal{L}$ $= \{L_{sel}, L_{prc}\}$ if the selector-component is from 
       $L_{sel}$ and each procedure-component is from $L_{prc}$. Note that a 
       member of $L_{prc}$ may appear zero, one, or more times in $S$. For
       example, both software systems in Fig. \ref{FigCSS}(d) are
       constructed using components drawn from $\cal{L}$ $= 
       \{L_{sel} = \{{\tt s1, s2}\}, L_{prc} = \{{\tt p1, p2, p3, p4}\}\}$.
\item A software system $S$ is derived from software system $S'$ relative
       to $\cal{L}$ $= \{L_{sel}, L_{prc}\}$ if there is a sequence of component
       and system adaptations relative to $\cal{L}$ that transforms $S$ into $S'$.
\item The operation of a software system $S$ satisfies requirements $R$ if
       for each situation-action pair $(s,a)$ in $R$, the execution of $S$
       relative to the truth-settings in $s$ results in the execution of $a$.
       For example, software system {\tt S1} in Fig. \ref{FigCSS}(d)
       satisfies the requirements in Fig. \ref{FigCSS}(a) but 
       software system {\tt S2} does not (because it produces different
       outputs ($a_3$, $a_1$, $a_1$, and $a_2$ respectively) for requirements $r_1$,
       $r_3$, $r_4$, and $r_5$).
\end{itemize}

\noindent
One's initial reaction on contemplating the above is that it is all way too simple ---
the radically restricted types of components and component composition considered
here can produce only the most basic memoryless reactive software systems, whose 
computational power is far less than that required in many real-world software systems. 
However, by the logic described previously in Section \ref{SectIntro} and discussed at
greater length in Section \ref{SectGenRes}, it is precisely the use of such simplified
models in the problems we analyze that will allow many of our derived results to apply
to a broad range of more realistic models of component-based software system development.

%\vspace*{-0.08in}

\subsection{Formalizations of computational problems}

\label{SectFormProb}

We can now formalize the problems sketched in Section \ref{SectFormIssues} as follows:

\vspace*{0.1in}

\noindent
{\sc Software Creation} ~ (SCre-Spec) \\
{\em Input}: Software requirements $R$ based on sets $I$ and $A$, a software-structure 
              specification $X = \langle I, A, |sel|, |prc|\rangle$. \\
{\em Output}: A software system $S$ whose structure is consistent with $X$
               and whose operation satisfies $R$, if such an $S$ exists, 
               and special symbol $\bot$ otherwise.

\vspace*{0.1in}

\noindent
{\sc Software Creation from Components}  ~ (SCre-Comp) \\
{\em Input}: Software requirements $R$ based on sets $I$ and $A$, a set of 
              software-component libraries $\cal{L}$ $= \{L_{sel}, L_{prc}\}$,
              a positive integer $d \geq 0$. \\
{\em Output}: A software system $S$ whose structure consists of 
               components based on at most $d$ types of components drawn from the 
               libraries in $\cal{L}$ and whose operation satisfies $R$, 
               if such an $S$ exists, and special symbol $\bot$ otherwise.

\vspace*{0.1in}

\noindent
{\sc Software Creation from Adapted Components} ~ (SCre-CompA) \\
{\em Input}: Software requirements $R$ based on sets $I$ and $A$, a set of 
              software-component libraries $\cal{L}$ $= \{L_{sel}, L_{prc}\}$,
              positive integers $d, c_c \geq 0$. \\
{\em Output}: A software system $S$ whose structure consists of 
               components derived by at most $c_c$ changes to at most $d$ 
               types of components drawn from the libraries in $\cal{L}$ and 
               whose operation satisfies $R$, if such an $S$ exists, and special
               symbol $\bot$ otherwise.

\vspace*{0.1in}

\noindent
{\sc Software Reconfiguration} ~ (SRec-Spec) \\
{\em Input}: A software system $S$ whose operation satisfies requirements $R$ and whose
              structure is consistent with specification $X = \langle I, A, |sel|, 
              |prc|\rangle$ , a set $R_{new}$ of new situation-action pairs
              based on $I$ and $A$, a positive integer $c_c \geq 0$. \\
{\em Output}: A software system $S'$ derived by at most $c_c$ code changes to $S$ 
               whose structure is consistent with $X$ and whose operation 
               satisfies $R \cup R_{new}$, if such an $S'$ exists, and special symbol
               $\bot$ otherwise.

\vspace*{0.1in}

\noindent
{\sc Software Reconfiguration from Components} ~ (SRec-Comp) \\
{\em Input}: A software system $S$ whose structure consists of 
              components from a set of software-component libraries
              $\cal{L}$ $= \{L_{sel}, L_{prc}\}$ and whose operation
              satisfies requirements $R$, a set $R_{new}$ of new situation-action pairs
              based on $I$ and $A$, positive integers 
              $c_l, d \geq 0$. \\
{\em Output}: A software system $S'$ derived from $S$ by at most $c_l$ 
               component changes relative to the libraries in $\cal{L}$
               whose structure consists of at most $d$ types of 
               components and whose operation satisfies $R \cup R_{new}$,
               if such an $S'$ exists, and special symbol $\bot$ otherwise.

\vspace*{0.1in}

\noindent
{\sc Software Reconfiguration from Adapted Components} ~ (SRec-CompA) \\
{\em Input}: A software system $S$ whose structure consists of 
              components from a set of software-component libraries
              $\cal{L}$ $= \{L_{sel}, L_{prc}\}$ and whose operation
              satisfies requirements $R$, a set $R_{new}$ of new situation-action pairs
              based on $I$ and $A$, positive integers 
              $c_l, c_c, d \geq 0$. \\
{\em Output}: A software system $S'$ derived from $S$ by at most $c_l$ 
               component changes 
               relative to the libraries in $\cal{L}$ and $c_c$ code changes 
               whose structure consists of at most $d$ types of 
               components and whose operation satisfies 
               $R \cup R_{new}$, if such an $S'$ exists, and special symbol $\bot$ 
               otherwise.

\vspace*{0.1in}

\noindent
We have included parameter $d$ in all component-based problems to allow control 
over and hence investigation of the computational effects of restricting the number of 
component retrievals from given libraries. In the case of problems SCre-Spec and
SRec-Spec where $d$ is not explicitly included, $d \leq |sel| + 2$ (as the number of
types of components in software systems created in these problems is bounded by the
number of available procedure-call slots in the selector). The inclusion of parameters 
$c_c$ and $c_l$ in the reconfiguration problems is justified analogously.

It is important to note the following three ways in which
these problems abstract away from practical aspects of CBD: 

\begin{enumerate}
\item the costs of searching for and accessing any component used to build a software 
       system is constant (as all components are in $L_{sel}$ and $L_{prc}$);
\item determining whether a set of components can be integrated 
       together into a running system (verification \cite[page 552]{MMM95})
       is trivial (as any selector from $L_{sel}$  with $k$ procedure-slots is 
       compatible with any selection of $k$ procedures from $L_{prc}$); and
\item determining whether a software system satisfies all of a given set of 
       requirements (validation \cite[page 552]{MMM95}) can be done in low-order 
       polynomial time (as each requirement in $R$ can be run checked against
       a two-level system $S$ in time proportional to the number of lines of
       code in the system, i.e., $(|sel| + 1)(|prc| + 2)$,
       to see if the requested action is produced).
\end{enumerate}

\begin{table*}[t]
\caption{Parameters for Component-based Software System Creation and Reconfiguration Problems. All 
          values in the fourth and fifth columns of the table are given for the first 
          and second solutions (namely, software systems s1 and s2) to the 
          creation problem posed at the end of Section \ref{SectFormEnt} relative
          to the requirements in part (a) of Fig. \ref{FigCSS}.
          }
\label{TabPrm}
\vspace*{0.1in}
\centering
\begin{tabular}{ | c ||  l |  c | c |  c |}
\hline
            &             &               & \multicolumn{2}{c}{Values} \\
Parameter   & Description & Applicability & S1 & S2 \\
\hline
\hline
$|I|$       & \# situation-variables & All & 5 & 5 \\
\hline
$|A|$       & \# possible actions & All & 3 & 3 \\
\hline
$|sel|$     & Max \# selector-conditions & All & 2 & 0 \\
\hline
$|prc|$     & Max \# procedure-conditions & All & 3 & 2 \\
\hline
$|S|$       & Max size of software system    & All & 15 & 4 \\
            & ($= (|sel| + 1)(|prc| + 2))$ &     &   &   \\
\hline
$d$     & Max \# component-types & All      & 4 & 2 \\
        & ~~ in software system  &        & &  \\
\hline\hline
$|L_{sel}|$ & \# selector-components & *-CompA, & 2 & 2 \\
            & ~~ in selector-library & *-Comp & & \\
\hline
$|L_{prc}|$ & \# procedure-components & *-CompA, & 4 & 4 \\
            & ~~ in procedure-library & *-Comp & & \\
\hline\hline
$|R_{new}|$   & \# new requirements & SRec-* & N/A & N/A \\
\hline
$c_c$         & Max \# code changes & *-CompA, & N/A & N/A \\
              & ~~ allowed         & SRec-Spec & & \\
\hline
$c_l$         & Max \# component changes & SRec-CompA, & N/A & N/A \\
              & ~~ allowed              & SRec-Comp & & \\
\hline
\end{tabular}
\end{table*}

\noindent
This is not to say that search, access, verification, and validation are unimportant; indeed,
all four of these activities are very difficult for all but the simplest software
systems and the subject of vigorous ongoing research \cite{GMS02,HC01,MMM95,VC+16}. 
Rather, we choose to abstract away from them so that our analyses can focus on the 
computational difficulties associated with the core activities in CBD, namely, the 
selection and adaptation of components.

\section{Component-based Software System Creation and \\ \hspace*{0.20in} 
          Reconfiguration are Intractable}

\label{SectCIRes}

Let us revisit the first of the questions raised in the Introduction --- namely,
are there efficient algorithms for component-based software system creation and 
reconfiguration that
are reliable, i.e., these algorithms operate correctly for all inputs? Following
general practice in Computer Science \cite{GJ79}, we shall say that an algorithm is 
efficient if that algorithm solves its associated problem in polynomial time --- that is,
the algorithm runs in time that is upper-bounded by $c_1n^{c_2}$ where $n$ is the input
size and $c_1$ and $c_2$ are constants. A problem which has a 
polynomial-time algorithm is {\em (polynomial-time) tractable}. Such algorithms are 
preferable to those whose runtimes are bounded by superpolynomial functions, e.g., 
$n^{\log n}$, $2^n$, $2^{2^n}$. This is so because polynomial functions increase in
value much slower than non-polynomial functions as $n$ gets large, which allows 
polynomial-time algorithms to solve much larger inputs in practical amounts of
time than superpolynomial-time algorithms.

Desirable as polynomial-time algorithms are, they do not exist
for any of our problems.\footnote{
All polynomial-time intractability results in this section hold relative to the
$P \neq NP$ conjecture, which is widely believed to be true \cite{For09,GJ79}.
}

\vspace*{0.15in}

\noindent
{\bf Result A}: SCre-Spec, SCre-Comp, SCre-CompA, SRec-Spec, SRec-Comp,
                 and SRec-CompA are not polynomial-time tractable.

\vspace*{0.15in}

\noindent
This shows that even the basic versions of component-based creation and reconfiguration
considered here are not solvable in polynomial time in general. A frequently-adopted 
response to this intractability is to relax reliability and consider probabilistic polynomial-time algorithms
which generate solutions that are of acceptable quality a very high proportion
of the time, e.g., genetic or simulated annealing algorithms \cite{HMZ12,LH+15,VGP08}. 
Unfortunately, this is not a viable option either.\footnote{
This result holds relative to both the $P \neq NP$ conjecture mentioned in Footnote 2 
and the $P = BPP$ conjecture, the latter of which is also widely believed to be 
true \cite{CRT98,Wig07}.
}

\vspace*{0.15in}

\noindent
{\bf Result B}: SCre-Spec, SCre-Comp, SCre-CompA, SRec-Spec, SRec-Comp,
                 and SRec-CompA are not polynomial-time tractable
                 by probabilistic algorithms which operate correctly with
                 probability $\geq 2/3$.

\vspace*{0.15in}

\noindent
In the next section, we will consider what may be a more useful approach --- namely,
algorithms that are reliable and run in what is effectively polynomial time
on restricted inputs encountered in practice.

\section{What Makes Component-based Software \\ \hspace*{0.2in} System Creation and 
          Reconfiguration \\ \hspace*{0,2in} Tractable?}

\label{SectPrmRes}

To answer the question of what restrictions make component-based software
system creation and 
reconfiguration tractable, we first need to define what it means to solve a 
problem efficiently under restrictions.  Let such restrictions be phrased in
terms of the values of aspects of our problem input or output; call each such aspect a 
\emph{parameter}. An overview of the parameters that we consider here is given in 
Table \ref{TabPrm}. These parameters can be divided into three groups:

\begin{enumerate}
\item Restrictions on software system, system requirement, and component structure 
       ($|I|$, $|A|$, $|sel|$, $|prc|$, $|S|$, $d$);
\item Restrictions on component library structure 
       ($|L_{sel}|$, $|L_{prc}|$);  and
\item restrictions on component adaptation ($|R_{new}|, c_c$, $c_l$).
\end{enumerate}

\noindent
Note that parameters $|S|$ and $c_c$ relative to problem SCre-CompA are of use in 
addressing the conjectures mentioned previously in Section \ref{SectIntro} about the 
effects of the extent and number of code modifications on the cost of adapting 
components.

The most popular conception of efficient solvability under restrictions 
phrased in terms of parameters is fixed-parameter tractability \cite{DF13}.
A problem $\Pi$ is {\em fixed-parameter (fp-) tractable 
relative to a set of parameters $K = \{k_1, k_2, \ldots, k_m\}$},
i.e., {\em $\langle K \rangle$-$\Pi$ is fp-tractable}, if there is an algorithm 
for $\Pi$ whose running time is upper-bounded by $f(K)n^c$ for some function
$f()$ where $n$ is the problem input size and $c$ is a constant. 
Note that fixed-parameter tractability generalizes polynomial-time solvability by 
allowing the leading constant $c_1$ in the runtime upper-bound of an algorithm to 
be a function of $K$ rather than a constant. 
This allow problems to be effectively solvable in 
polynomial time when the values of the parameters in $K$ are small and $f()$
is well-behaved, e.g., $1.2^{k_1 + k_2}$, such that
the value of $f(K)$ is a small constant. Hence, if a polynomial-time 
intractable problem $\Pi$ is fixed-parameter tractable for a parameter-set 
$K$ when the values of the parameters in $K$ are small and $f()$ is 
well-behaved then $\Pi$ can be efficiently solved even for large inputs.

Our questions about efficient solvability of component-based software creation and 
reconfiguration problems under restrictions may now be rephrased in terms of
what sets of parameters do and do not make these problems fp-tractable. 
%the analyses in these subsections will focus on 
%characterizing
%the complexities of the largest possible number of parameter-combinations relative to 
%the parameters in Table \ref{TabPrm} rather than deriving the best possible algorithms 
%for parameter-sets that yield tractability.
We consider first what parameters do not yield fp-tractability.\footnote{
Each fp-intractability results in this section holds relative to the conjecture 
$FPT \neq W[1]$, which is widely believed to be true \cite{DF13}.
}

\vspace*{0.15in}

\noindent
{\bf Result C}: $\langle |A|, |sel|, |prc|, |S|\rangle$-SCre-Spec is fp-intractable.

\vspace*{0.15in}

\noindent
{\bf Result D}: $\langle |A|, |prc|, d, |L_{sel}|\rangle$-SCre-Comp is fp-intractable .

\vspace*{0.15in}

\noindent
{\bf Result E}: $\langle |A|, |sel|, |prc|, |S|, d, |L_{sel}|, |L_{prc}|, 
                 c_c \rangle$-SCre-CompA is fp-intractable.

\vspace*{0.15in}

\noindent
{\bf Result F}: $\langle |A|, |sel|, |prc|, |S|, |R_{new}|, c_c\rangle$-SRec-Spec is 
                 fp-intractable.

\vspace*{0.15in}

\noindent
{\bf Result G}: $\langle |A|, |prc|, d, |L_{sel}|\rangle$-SRec-Comp is fp-intractable.

\vspace*{0.15in}

\noindent
{\bf Result H}: $\langle |A|, |sel|, |prc|, |S|, d, |L_{sel}, |L_{prc}|, |R_{new}|, c_l, 
                 c_c \rangle$-SRec-CompA is \newline fp-intractable.

\vspace*{0.15in}

\noindent
These results are more powerful than they first appear, as a problem that is 
fp-intractable relative to a particular parameter-set $K$ is also fp-intractable 
relative to any subset of $K$ \cite[Lemma 2.1.31]{War99}. Hence, there are in fact a 
number of combinations of parameters which cannot be restricted to yield 
fp-tractability for our problems. Courtesy of Result E, this includes parameters $|S|$ 
and $c_c$ relative to problem SCre-CompA.

That being said, there are restrictions that do make our problems fp-tractable. 

\vspace*{0.15in}

\noindent
{\bf Result I}: $\langle I\rangle$-SCre-Spec, -SCre-Comp, -SCre-CompA, -SRec-Spec, 
                    -SRec-Comp, and -SRec-CompA are fp-tractable.

\vspace*{0.15in}

\noindent
{\bf Result J}: $\langle |sel|, |L_{prc}|\rangle$-SCre-Comp and -SRec-Comp are 
                    fp-tractable.

\vspace*{0.15in}

\noindent
Again, these results are more powerful than they first appear, as 
a problem that is fp-tractable relative to a particular parameter-set $K$ 
is also fp-tractable relative to any superset of $K$ \cite[Lemma 2.1.30]{War99}.
Hence, any set of parameters including $|I|$ can be restricted to yield fp-tractability
for all of our problems and any set of parameters including both $|sel|$ and
$|L_{prc}|$ can be restricted to yield fp-tractability for SCre-Comp and SRec-Comp.

A detailed summary of our fixed-parameter intractability and tractability results is 
given in Table \ref{TabFPRes}. Note that our fp-intractability results hold when many 
of the parameters in these results are restricted to small constant values. Moreover, 
courtesy of the parameter-set subset/superset relations noted above, Results C, E, F, 
H, and I completely characterize the parameterized complexity of problems SCre-Spec, 
SCre-CompA, SRec-Spec, and SRec-CompA relative to the parameters in Table \ref{TabPrm}. 

\begin{table}[p]
\caption{A Detailed Summary of Our Fixed-Parameter Complexity Results. Each column 
          in this table is a result for a particular problem which holds relative to
          the parameter-set consisting of all parameters with a $@$-symbol in that column.
	  If in addition  a result holds when a parameter has a constant value $c$, that is
	  indicated by $c$ replacing $@$.
          Results are grouped by problem, with fp-intractability results
          first and fp-tractability results (shown in bold) last.}
\label{TabFPRes}
\begin{center}

\vspace*{0.15in}

%\begin{tabular}{ l  l  l  l  l  l  l  l }
\begin{tabular}{ | l ||  l  l | l  l  l |  l  l | }
\hline
      & \multicolumn{7}{c|}{Creation} \\
\cline{2-8}
      & \multicolumn{2}{l|}{Spec} 
      & \multicolumn{3}{l|}{Comp} 
      & \multicolumn{2}{l|}{CompA} \\
      & C & {\bf I}
      & D & {\bf I} & {\bf J}
      & E & {\bf I} \\
\hline\hline
$|I|$       & -- & {\bf @}
            & -- & {\bf @} & --
            & -- & {\bf @} \\
\hline
$|A|$       & 2 & --
            & 2 & -- & --
            & 2 & -- \\
\hline
$|sel|$     & 0 & --
            & -- & -- & {\bf @}
            & 0 & -- \\
\hline
$|prc|$     & @ & --
            & 1 & -- & --
            & @ & -- \\
\hline
$|S|$     & @ & --
          & -- & -- & --
          & @ & -- \\
\hline
$d$         & 2 & --
            & @ & -- & --
            & 2 & -- \\
\hline\hline
$|L_{sel}|$ & N/A & N/A
            & 1 & -- & --
            & 1 & -- \\
\hline
$|L_{prc}|$ & N/A & N/A
            & -- & -- & {\bf @}
            & 1 & -- \\
\hline\hline
$|R_{new}|$ & N/A & N/A
            & N/A & N/A & N/A
            & N/A & N/A \\
\hline
$c_c$       & N/A & N/A
            & N/A & N/A & N/A
            & @ & -- \\
\hline
$c_l$       & N/A & N/A
            & N/A & N/A & N/A
            & N/A & N/A \\
\hline
\end{tabular}

\vspace*{0.25in}

\begin{tabular}{ | l ||  l  l | l  l  l | l  l | }
\hline
      & \multicolumn{7}{c|}{Reconfiguration} \\
\cline{2-8}
      & \multicolumn{2}{l|}{Spec} 
      & \multicolumn{3}{l|}{Comp} 
      & \multicolumn{2}{l|}{CompA} \\
      & F & {\bf I}
      & G & {\bf I} & {\bf J}
      & H & {\bf I} \\
\hline\hline
$|I|$       & -- & {\bf @}
            & -- & {\bf @} & --
            & -- & {\bf @} \\
\hline
$|A|$       & 3 & --
            & 2 & -- & --
            & 3 & -- \\
\hline
$|sel|$     & 0 & --
            & -- & -- & {\bf @}
            & 0 & -- \\
\hline
$|prc|$     & @ & --
            & 1 & -- & --
            & @ & -- \\
\hline
$|S|$       & @ & --
            & -- & -- & --
            & @ & -- \\
\hline
$d$         & 2 & --
            & @ & -- & --
            & 2 & -- \\
\hline\hline
$|L_{sel}|$ & N/A & N/A
            & 1 & -- & --
            & 1 & -- \\
\hline
$|L_{prc}|$ & N/A & N/A
            & -- & -- & {\bf @}
            & 1 & -- \\
\hline\hline
$|R_{new}|$ & 1 & --
            & -- & -- & --
            & 1 & -- \\
\hline
$c_c$       & @ & --
            & N/A & N/A & N/A
            & @ & -- \\
\hline
$c_l$       & N/A & N/A
            & -- & -- & --
            & 0 & -- \\
\hline
\end{tabular}
\end{center}
\end{table}

\section{Generality of Results}

\label{SectGenRes}

All of our intractability results (namely, Results A--H), though defined 
relative to basic models of software systems, system requirements, components, 
component libraries, and component adaptation, have a broad applicability. This
is so because, as noted in Section 1, 
the models for which these results hold are special cases of more realistic 
models, e.g.,

\begin{itemize}
\item purely functional software requirements that explicitly list all situations to 
       which software should respond are a special case of more complex requirements
       based on compact specifications of software behaviour such as finite-state 
       automata, statecharts, or statements in natural language which in addition 
       incorporate other functional and non-functional properties of software systems 
       such as degree of reliability and response time;
\item two-level selector / procedure software systems that act as
       simple functions are a special case of more complex 
       persistent software systems 
       whose components invoke each other in more complex 
       manners;
\item components consisting of a single condition-statement block 
       procedure without input parameters or return values and which have
       no dependencies on other components are a special case of more 
       complex components consisting of arbitrary procedure code and/or data types 
       which have complex dependencies on other components such as
       data-type sharing or inheritance;
\item component composition by calls to procedures without input parameters or return 
       values and which have no dependencies on other components is a special case of 
       more complex types of component composition invoking complex dependencies on other components such as
       data-type sharing or inheritance;
\item component libraries that are simply lists of components are
       a special case of component libraries that incorporate 
       component behaviour-specifications and / or metadata
       to aid in selection and adaptation; and
\item whitebox component adaptation involving single-line single-symbol code changes 
       is a special case of whitebox component adaptation involving more complex
       types of code change.
\end{itemize}

\noindent
More realistic versions of the six component-based software system creation and 
reconfiguration problems analyzed in this paper can be created by replacing any
combination of simple special-case models with the more realistic models above. For 
example, one could define a version of SCre-Comp in which software requirements are
given as natural language statements and components consist of arbitrary internal
code composed with each other by complex interfaces such as those allowed in
popular component models like CORBA. Courtesy of the special-case relationship, 
any automated system that solves such a more 
realistic problem $\Pi'$ can also solve the problem $\Pi$ defined relative
to the simple special-case models analyzed here. Depending on the nature
of the special-case relationship, this will sometimes involve recoding problem 
entities, e.g.,  translating a software requirement $((i_1 = True, i_2 = False, i_3 =
True), a_4)$ in $R$ into a natural language statement ``If $i_1$ is $True$ and $i_2$ 
is $False$ and $i_3$ is $True$ then execute action $a_4$''. Intractability results for
$\Pi$ must then also apply to $\Pi'$ as well as the operation of any automated system 
solving $\Pi'$. To see this, suppose $\Pi$ is intractable; if $\Pi'$ is tractable 
by algorithm $A$, then $A$ can be used to solve $\Pi$ efficiently, which contradicts 
the intractability of $\Pi$. Hence, $\Pi'$ must also be intractable.

Our various tractability results under restrictions (namely, Results I and J) are
much more fragile. This is so because what may seem like innocuous changes to the 
problems analyzed here can violate assumptions critical to the correct operation of 
the algorithms underlying these results, which means that algorithms for solving
simplified problems may not work for more realistic problems. Hence, our tractability
results may only apply to certain more realistic problems, and this needs
to be assessed on a case-by-case basis. 

\section{Discussion}

\label{SectDisc}

We have found that it is very unlikely that creating and reconfiguring software systems
by {\em de novo} design, component selection, and component selection with
adaptation are polynomial-time solvable for even the simple component-based
software systems considered here by either deterministic (Result A)
or probabilistic (Result B) algorithms. This answers the long-standing question of which
of {\em de novo} design, component selection, or component selection with
adaptation is best for creating software systems --- computationally speaking,
all three methods are equally good (or bad) in general. Results A 
and B also constitute the first proof that {\em no} currently-used method (including 
search-based methods based on evolutionary algorithms (see \cite[Section 5]{HMZ12} and 
references)) can guarantee both efficient and (even high-probability) correct operation 
for all inputs for these problems 

As described in Section \ref{SectPrmRes}, efficient correctness-guaranteed methods
may yet exist for these problems relative to plausible restrictions. It seems reasonable
to conjecture that restrictions relative to the parameters listed in Table \ref{TabPrm}
should render our problems tractable. We have shown that many of these restrictions 
(either individually or in combination) do not yield fixed-parameter tractability 
(Results C--H), even when many of the parameters involved are 
restricted to very small constants (see Table \ref{TabFPRes}). That
being said, we do have some initial fixed-parameter tractability results (Results I and J).
Taken together, these sets of results have the following 
immediate implications:

\begin{itemize}
\item Our results fully characterize the fixed-parameter complexity of problems 
       SCre-Spec, SCre-CompA, SRec-Spec, and SRec-CompA relative to the parameters in
       Table \ref{TabPrm}. It is interesting that the only parameter considered
       that matters with respect to fixed-parameter tractability
       is the environment which can be sensed by the software system ($|I|$ (Result I)).
\item Our results partially characterize the fixed-parameter complexities of problems 
       SCre-Comp and SRec-Comp. If fp-tractability holds in at least some of those 
       uncharacterized cases, this would constitute proof that creating and 
       reconfiguring software systems by component selection is tractable in a wider 
       set of circumstances than (and hence may be preferable to) creating and 
       reconfiguring software systems by either {\em de novo} design or component 
       selection with adaptation. 
\item Result E establishes that SCre-CompA is not fp-tractable relative to
       either parameter $|S|$ (which measures the amount of component code that can be 
       modified) or parameter $c_c$ (which measures the number of code modifications).
       This suggests that explanations of good performance by human programmers in
       component selection by whitebox adaptation which are based on small values of
       these parameters \cite[Page 537]{MMM95} are incomplete, and that
       other parameters must also be involved and of small value (if humans are invoking 
       fp-tractable algorithms). 
       This is worth resolving, both for suggesting how human programmers can do CBD
       more productively as well as characterizing which parameters considered here are
       and are not of small value with respect to real-world CBD development.
\item Results F and H show that SRec-Spec and SRec-CompA
       are not fixed-parameter tractable even when there is only a single
       new system requirement, i.e., $|R_{new}| = 1$. That it only takes one new
       requirement to cause intractability for these problems and that this cannot be 
       mitigated by invoking fixed-parameter tractability relative to 
       $|R_{new}|$ is unsettling. The situation, however, is still open with respect
       to this parameter for SRec-Comp, and may yet yield another circumstance 
       under which software system reconfiguration by component selection is preferable
       to software system reconfiguration by either {\em de novo} design or
       component selection with adaptation.
\end{itemize}

\noindent
Depending on questions of interest, other implications may also be gleaned from our
results. In any case, observe that many of these implications are only possibly by 
virtue of using explicit and detailed formalizations of software systems, system 
requirements, components, and component adaptation, and hence could not have been 
inferred from previous complexity analyses establishing the general intractability of 
software system creation by both component selection \cite{PO99,PWM03} and component selection 
with adaptation \cite{BBR05}. This demonstrates that, as promised at the end of Section
\ref{SectFormIssues}, our analyses (and in particular, our parameterized analyses) do 
indeed build on and extend the results derived in those previous analyses.

It is important to note that the brute-force search algorithms underlying all of
our tractability results are not immediately usable in real-world software engineering 
because (1) the problems solved by these algorithms are simplified versions of 
real-world problems and (2) the running times of these algorithms are (to be blunt, 
ludicrously) exorbitant. The first difficulty is a consequence of our complexity 
analysis methodology described in Section \ref{SectIntro} and cannot be avoided.
This is not so bad, as algorithms for simplified problems may still be useful guides in
developing algorithms for more complex and realistic versions of those problems. The 
second difficulty is also not as bad as it initially seems because it is well known within
the parameterized complexity research community that once fp-tractability is 
proven relative to a parameter-set, surprisingly effective fp-algorithms can often be 
subsequently developed \cite{CF+15, DF13}. This may involve
adding additional parameters to ``minimal'' parameter-sets that imply the greatest
possible number of fp-tractability results. Such algorithms typically have runtimes
with greatly diminished non-polynomial terms and polynomial terms that are are additive
rather than multiplicative and linear in the input size. 

The cautionary note above highlights a more general caveat --- namely, that the 
analyses in this paper are only intended to sketch out relative to which sets of 
parameters efficient algorithms may and may not exist for our problems of 
interest and are not intended to be of immediate use in real-world CBD. Indeed, given
the successes in applying CBD over the last several decades in various real-world
settings, one may even argue that our results are irrelevant --- if we have methods
that work in practice, who needs theory? We believe such a view is short-sighted at
best and dangerous at worst. Not knowing the precise conditions under which existing
CBD methods work well may have serious consequences, e.g., drastically slowed 
software creation time and/or unreliable software operation, if these conditions are 
violated. These consequences would be particularly damaging in the case of fully 
automatic applications like ubiquitous and autonomic computing. Given that reliable 
software operation is crucial and efficient software creation and reconfiguration is at
the very least desirable, the acquisition of such knowledge via a combination of 
rigorous empirical and theoretical analyses should be a priority. With respect to 
theoretical analyses, it is our hope that the techniques and results in this 
paper comprise a useful first step.

\section{Conclusions}

\label{SectConc}

We have presented a formal characterization of the problems of automatic software 
system creation
and reconfiguration by {\em de novo} design, component selection, and component 
selection with adaptation. Our complexity analyses reveal that, while all of these 
approaches are computationally intractable in general relative to both deterministic 
and probabilistic algorithms, there are restrictions that render each of these approaches
tractable. The results of these analyses give the first rigorous computational comparison
between the three approaches to software design considered here. In particular, these 
results establish that all three approaches are equally computationally difficult in 
general but suggest that software system creation and reconfiguration by component 
selection may be tractable under more circumstances than software system creation and 
reconfiguration relative to the other two approaches.

There are several promising directions for future research:

\begin{enumerate}
\item Complete the fixed-parameter analyses initiated in this paper
       relative to the parameters considered here.
\item Define and analyze CBD-related problems that incorporate aspects
       of CBD that were abstracted away in the problems considered here, e.g.,
       component search and component-set verification (Section \ref{SectFormProb}).
\item Define and analyze CBD-related problems that are not special cases of those
       considered here, e.g., CBD by blackbox component
       adaptation using adaptors \cite{MA04,Kel08}.
\item Consider the computational complexity of CBD relative to additional types of 
       tractability, e.g., polynomial-time approximation \cite{AC+99} and 
       fixed-parameter approximation \cite{Mar08}, probabilistic  
       \cite{Kwi15}, and evolutionary \cite{KN13} algorithms.
\end{enumerate}
 
\noindent
It is our hope that this research will, in time, be of use in both deriving the best 
possible implementations of fully automatic component-based software system creation 
and reconfiguration methods for real-world CBD and inspiring additional
parameterized complexity analyses of other activities in software engineering, e.g.,
automatic software system (re)modularization (see \cite{War15} and references).

\section*{Acknowledgments}
The authors extend most grateful thanks to Iris van Rooij and Maria
Otworowska, with whom they collaborated on the formalization 
of the adaptive toolbox theory of human decision-making that underlies
the formalizations of component-based software system creation and 
reconfiguration analyzed here. We also want to thank Antonina Kolokolova and 
Johan Kwisthout for useful discussions on Rice's Theorem and efficient 
probabilistic solvability, respectively.
Funding for this work was provided by a National Science and Engineering
Research Council (NSERC) Discovery Grant to TW (grant \# 228105-2015).


\appendix

\section{Proofs of Results}

\label{SectProof}

All of our intractability results will be derived using polynomial-time and 
parameterized reductions\footnote{
Given two problems $\Pi$ and $\Pi'$, a polynomial-time reduction from $\Pi$
to $\Pi'$ is essentially a polynomial-time algorithm for transforming instances
of $\Pi$ into instances of $\Pi'$ such that any polynomial-time algorithm for $\Pi'$ 
can be used in conjunction with this instance transformation algorithm to create
a polynomial-time algorithm for $\Pi$. Analogously, a
parameterized reduction from $\langle K \rangle$-$\Pi$ to $\langle K'\rangle$-$\Pi'$
allows the instance transformation algorithm to run in fp-time relative $K$ and
requires that for each $k' \in K'$ that there is a function $g_{k'}()$ such that
$k' = g_{k'}(K)$. Such an instance transformation algorithm can be used in conjunction
with any fp-algorithm for $\langle K'\rangle$-$\Pi'$ to create an fp-algorithm for
$\langle K\rangle$-$\Pi$. The interested reader is referred to \cite{DF13,GJ79}
for details.
}
from the following problem:

\vspace*{0.1in}

\noindent
{\sc Dominating Set} \\
{\em Input}: An undirected graph $G = (V, E)$ and an integer $k$. \\
{\em Question}: Does $G$ contain a dominating set of size $\leq k$, i.e., is
              there a subset $V' \subseteq V$, $|V'| \leq k$, such that for
	      all $v \in V$, either $v \in V'$ or there is a $v' \in V'$ such
	      that $(v, v') \in E$?

\vspace*{0.1in}

\noindent 
This problem is $NP$-hard in general \cite[Problem GT2]{GJ79} and $W[2]$-hard relative 
to parameter-set $\{k\}$ \cite{DF13}.
For each vertex $v \in V$, let the complete neighbourhood $N_C(v)$ of $v$ be the
set composed of $v$ and the set of all vertices in $G$ that are adjacent to $v$ 
by a single edge, i.e., $v \cup \{ u ~ | ~ u ~ \in V ~ \rm{and} ~ (u,v) \in E\}$.
We assume below an arbitrary ordering
on the vertices of $V$ such that $V = \{v_1, v_2, \ldots, v_{|V|}\}$.

For technical reasons, all intractability results are proved relative to decision
versions of problems, i.e., problems whose solutions are either ``yes'' or ``no''.
Though none of our problems defined in Section 2.3 of the main text
are decision problems, each can be made into a decision problem by asking if that 
problem's requested output exists; let that decision version for a problem {\bf X}
be denoted by {\bf X}${}_D$. The following easily-proven lemma will be useful
below in transferring results from decision problems to their associated
non-decision problems; it follows from the observation that
any algorithm for non-decision problem {\bf X} can be used to solve {\bf X}${}_D$.

\begin{lemma}
If {\bf X}$_D$ is $NP$-hard then {\bf X} is not solvable in polynomial time
unless $P = NP$.
\label{LemProp1}
\end{lemma}

\begin{lemma}
{\sc Dominating Set} polynomial-time many-one reduces to SCre- \linebreak Spec$_D$ such that in
the constructed instance of SCre-Spec$_D$, $|A| = d = 2$, $|sel| = 0$,  and $|prc|$ and
$|S|$ are functions of $k$ in the given instance of {\sc Dominating Set}.
\label{LemRedSCreSpec}
\end{lemma}

\begin{proof}
Given an instance $\langle G = (V, E), k\rangle$ of {\sc Dominating 
set}, the constructed instance $\langle I, A, R, X = 
\langle |sel|,|prc|\rangle \rangle$ of SCre-Spec$_D$ has $I = \{i_1, i_2,
\ldots i_{|V|}\}$, i.e., there is a unique situation-variable corresponding
to each vertex in $V$, $A = \{0, 1\}$, $|sel| = 0$, and $|prc| = k$. 
As $|sel| = 0$ and for this problem $d = |sel| + 2$, $d = 2$. There
are $|V| + 1$ situation-action pairs in $R$ such that (1) for $r_i = 
(s_i, a_i)$, $1 \leq i \leq |V|$, $v(i_j) = True$ if $v_j \in N_C(v_i)$
and is $False$ otherwise and $a_i = 1$, and (2) for $r_{|V|+1} =
(s_{|V|+1}, a_{|V|+1})$, all $v(i_j)$ are $False$ and $a_{|V|+1} = 0$.
Note that the instance of SCre-Spec$_D$ described above can be constructed in 
time polynomial in the size of the given instance of {\sc Dominating Set}. 

If there is a dominating set $D$ of size at most $k$ in the given instance of
{\sc Dominating Set}, construct a software system consisting of a 
default \texttt{*}-condition selector and a single
procedure in which the $k$ \texttt{IF-THEN} statements have 
situation-variable conditions corresponding to the vertices in $D$
and action $1$ and the final 
\texttt{ELSE} statement has action $0$. Observe that this
software system satisfies all situation-action pairs in $R$ and has
$|sel| = 0$ and $|prc| \leq k$. 

Conversely, suppose that the constructed instance 
of SCre-Spec$_D$ has a software system satisfying $R$ with $|sel| = 0$  and
$|prc| \leq k$. The selector in this system must be the default selector
as it is the only selector with $|sel| = 0$. As for the
single procedure attached to that selector, it has two possible
structures:

\begin{enumerate}
\item{
\textit{No negated situation-variable conditions}: As $r_{|V|+1}$ has
no situation- \linebreak variables set to $True$, it can only be processed correctly
by the final \texttt{ELSE}-statement, which must thus execute 
action $0$. In order to accept all other situation-action pairs in $R$,
the remaining $\leq k$ {\tt IF-THEN} statements must both have associated executed action 1 and
situation-variables whose corresponding vertices form a dominating set in $G$
of size at most $k$.
}
\item{
\textit{Negated situation-variables are present}: 
Let $c$ be the first
negated situation-variable condition encountered moving down the
code in the procedure, $C$ be the set of unnegated situation-variable
conditions encountered before $c$, and $R' \subseteq R$ be the set of
situation-action pairs in $R$ that are not recognized by the situation-variables in $C$.
As $|prc| \leq k$, $|C| \leq k - 1$. As
situation-action pair $r_{|V|+1}$ has no 
situation-variable with value $True$ and hence cannot be recognized
by an \texttt{IF-THEN} statement with an unnegated situation-variable
condition, $r_{|V|+1} \in R'$. Moreover, as each situation-action pair
in $R - R'$ has associated action 1, all actions executed by
the {\tt IF-THEN} statements associated with the variables in $C$ must be 1.

Let $R'' \subset R$ be
such that $R' = R'' \cup \{r_{|V|+1}\}$. If $R''$ is empty, then the 
variables in $C$ must have recognized all situation-action pairs in
$R$ except $r_{|V|+1}$, and hence the vertices corresponding to the
situation-variables in $C$ form a dominating set in $G$ of size at
most $k - 1$. If $R''$ is not empty, the situation-variable in $c$ cannot 
have value $False$ for any of the situation-action pairs in $R''$
because having either $0$ or $1$ as the action executed by the 
associated \texttt{IF-THEN} statement would result in the procedure executing 
the wrong action for at least one situation-action pair in $R'$ (if $0$, at least
one of the situation-action pairs in $R''$; if $1$, $r_{|V|+1}$). However,
this implies that all situation-action pairs in $R''$ have value $True$
for the situation- variable in $c$, which in turn implies that the vertices 
corresponding to the situation-variables in $C \cup \{c\}$ form a dominating 
set in $G$ of size at most $k$.
}
\end{enumerate}

\noindent
Hence, the existence of a satisfying software system for the
constructed instance of SCre-Spec$_D$ implies the existence of a
dominating set of size at most $k$ for the given instance  of
{\sc Dominating Set}.

To complete the proof, note that in the constructed instance of \linebreak 
SCre-Spec$_D$, $|A| = d = 2$, $|sel| = 0$, $|prc| = k$, and $|S| = (|sel| + 1)(|prc| + 2) =
k + 2$.
\end{proof}

\begin{lemma}
{\sc Dominating Set} polynomial-time many-one reduces to \linebreak SCre-CompA$_D$ 
such that in the constructed instance of SCre-CompA$_D$, $|A| = d = 2$, $|sel| = 0$, 
$|L_{sel}| = L_{prc}| = 1$, and $|prc|$, $|S|$, and $c_c$ are functions of $k$ in 
the given instance of {\sc Dominating Set}.
\label{LemRedSCreCompA} 
\end{lemma} 

\begin{proof}
Given an instance $\langle G = (V, E),$ $k\rangle$ of {\sc Dominating Set}, 
the constructed instance $\langle I, A, R,$ $L_{sel},$ $L_{prc}, d, c_c \rangle$
of SCre-CompA$_D$ has $I$, $A$, and $R$ as in the reduction in Lemma 
\ref{LemRedSCreSpec} above. Let $L_{sel}$ consist of the default selector and 
$L_{prc}$ consist of a single procedure whose $k$ \texttt{IF-THEN} statements 
have conditions that are arbitrarily-selected unnegated situation-variables in $I$ and 
execute action $1$ and whose \texttt{ELSE} statement executes action $0$. 
Finally, set $d = 2$ and $c_c = k$.
Note that the instance of SCre-CompA$_D$ described above can be constructed in 
time polynomial in the size of the given instance of {\sc Dominating Set}. 

If there is a dominating set $D$ of size at most $k$ in the given instance of
{\sc Dominating Set}, construct the software system $S'$ consisting of the single 
selector and procedure given in libraries $L_{sel}$ and $L_{prc}$,
respectively, and create $S$ from $S'$ by modifying the first $|D|$ \texttt{IF-THEN} 
statements in the procedure to have situation-variable conditions corresponding to 
the vertices in $D$ and modifying the remaining $k - |D|$ \texttt{IF-THEN} statements
in the procedure to have situation-variable conditions corresponding to an arbitrary
subset of $D$.  Observe that $S$ satisfies all situation-action pairs in $R$, has
$|sel| = 0$ and $|prc| = k$, and is obtained by at most $k$ modifications to the
system code of $S'$.

Conversely, suppose that the constructed instance 
of SCre-CompA$_D$ has a software system $S$ satisfying $R$ that was obtained by at
most $c_c$ modifications to the system code of a software system $S'$ consisting
of at most $d$ components drawn from $L_{sel}$ and $L_{prc}$. There is exactly
one such $S'$ consisting of the single selector and procedure in libraries
$L_{sel}$ and $L_{prc}$, respectively. The modifications made to the procedure
in $S'$ to create $S$ result in one of the two structures described in the proof of 
correctness of the reduction in Lemma \ref{LemRedSCreSpec} above, both of which (as 
shown in that same proof) allow the derivation of dominating sets in $G$ of size at 
most $k$.

To complete the proof, note that in the constructed instance of \linebreak
SCre-CompA$_D$, $|A| = d = 2$, $|sel| = 0$, $|L_{sel}| = |L_{prc}| = 1$, 
$|prc| = c_c = k$, and $|S| = (|sel| + 1)(|prc| + 2) = k + 2$.
\end{proof}

\vspace*{-0.1in}

\begin{lemma}
{\sc Dominating Set} polynomial-time many-one reduces to SCre- \linebreak Comp$_D$ 
such that in the constructed instance of SCre-Comp$_D$, $|A| = 2$, $|prc| = |L_{sel}| = 1$,
and $d$ is a function of $k$ in the given instance of {\sc Dominating Set}.
\label{LemRedSCreComp} 
\end{lemma} 

\begin{proof}
Given an instance $\langle G = (V, E), k\rangle$ of {\sc Dominating 
set}, the constructed instance $\langle I, A, R, L_{sel}, L_{prc},
d\rangle$ of SCre-Comp$_D$ has $I = \{i_1, i_2,
\ldots i_{2|V|}\}$ and $A = \{0, 1\}$. There
are $|V| + 1$ situation-action pairs in $R$ such that 
(1) for $r_i = 
(s_i, a_i)$, $1 \leq i \leq |V|$, $v(i_i) = True$, $v(i_{|V| + j}) = True$ 
if $v_j \in N_C(v_i)$, all other $v(i_j)$ are $False$,
and $a_i = 1$, and (2) for $r_{|V|+1} =
(s_{|V|+1}, a_{|V|+1})$, all $v(i_j) = False$ and $a_{|V|+1} = 0$.
There is a single selector in $L_{sel}$ whose $|V| - 1$ \texttt{IF-THEN}
statement conditions are the first $|V| - 1$ situation-variables in $I$.
There are $|V|$ procedures in $L_{prc}$ such that procedure $j$,
$1 \leq j \leq |V|$, has an \texttt{IF-THEN} statement 
whose condition is situation-variable $i_{|V| + j}$
and which executes action $1$ and an \texttt{ELSE} statement that
executes action $0$. As such, each procedure $j$ in $L_{prc}$ effectively 
corresponds to and encodes vertex $v_j$. Finally, let $d = k + 1$. 
Note that the instance of SCre-Comp$_D$ described above can be constructed in 
time polynomial in the size of the given instance of {\sc Dominating Set}. 

If there is a dominating set $D$ of size at most $k$ in the given instance
of {\sc Dominating Set}, construct a software system $S$ consisting of the 
single selector in $L_{sel}$ and the set $X$ consisting of the procedures 
in $L_{prc}$ corresponding to the vertices in $D$. For the $i$th \texttt{IF-THEN} 
statement in the selector, $1 \leq i \leq |V| - 1$, execute any of the procedures 
in $X$ encoding a vertex in $D$ that dominates $v_i$. For the 
\texttt{ELSE} statement in the selector, execute any of the procedures in $X$
encoding a vertex in $D$ that dominates vertex $v_{|V|}$. Observe that 
$S$ satisfies all situation-action pairs in $R$ (with both $r_{|V|}$ and
$r_{|V|+1}$ being satisfied by the procedure associated with the 
\texttt{ELSE} statement in the selector) and has $d \leq k + 1$. 

Conversely, suppose that the constructed instance 
SCre-Comp$_D$ has a software system $S$ constructed from at most $d = k + 1$
distinct components from $L_{sel}$ and $L_{prc}$ that satisfies all
of the situation-action pairs in $R$. As the selector is constructed 
such that each of the first $|V| - 1$ situation-action pairs in $R$ has
its own associated procedure and each procedure accepts a
specific vertex in $G$, in order for these situation-action 
pairs to be satisfied, the distinct procedures in $S$
must correspond to a dominating set for $G$ (note that both $r_{|V|}$ and
$r_{|V|+1}$ are satisfied by the procedure associated with the 
\texttt{ELSE} statement in the selector). As $d \leq k + 1$ and $S$
had to incorporate a selector from $L_{sel}$, 
there are at most $k$ such procedures and hence there is a dominating set
of size at most $k$ in the given instance of {\sc Dominating Set}.

To complete the proof, note that in the constructed instance of \linebreak
SCre-Comp$_D$, $|A| = 2$, $|prc| = |L_{sel}| = 1$, and $d = k + 1$.
\end{proof}

\vspace*{0.06in}

\begin{lemma}
{\sc Dominating Set} polynomial-time many-one reduces to SRec- \linebreak Spec$_D$ such that in
the constructed instance of SRec-Spec$_D$, $|A| = 3$, $|sel| = 0$, $d = 2$, $|R_{new}| = 1$, and 
$|prc|$, $|S|$, and $c_c$ are functions of $k$ in the given instance of {\sc Dominating Set}.
\label{LemRedSRecSpec}
\end{lemma}

\begin{proof}
Given an instance $\langle G = (V, E), k\rangle$ of {\sc Dominating 
set}, the constructed instance $\langle I, A, R, S, X = \langle |sel|,
|prc|\rangle, R_{new}, c_c \rangle$ of SRec-Spec$_D$ has $I =$ \linebreak $ \{i_1, i_2,
\ldots i_{|V|+1}\}$, $A = \{0, 1, 2\}$, $|sel| = 0$, and $|prc| = k + 1$. 
As $|sel| = 0$ and for this problem $d = |sel| + 2$, $d = 2$. There
are $|V| + 1$ situation-action pairs in $R$ such that (1) for $r_i = 
(s_i, a_i)$, $1 \leq i \leq |V|$, $v(i_j) = True$ if $v_j \in N_C(v_i)$,
$v(i_{|V|+1}) = True$, all remaining $v(i_j)$ are $False$, and $a_i = 1$, 
and (2) for $r_{|V|+1} = (s_{|V|+1}, a_{|V|+1})$, all $v(i_j)$ are $False$ and 
$a_{|V|+i} = 0$. Software system $S$ consists of the default selector and a single 
procedure based on $k + 1$ conditions, in which the first $k$ \texttt{IF-THEN} 
statements have arbitrary conditions drawn from the first $|V|$ situation-variables 
in $I$ and execute action $1$, the $(k + 1)st$ \texttt{IF-THEN} statement has
condition $i_{|V|+1}$ and executes action $1$, and the \texttt{ELSE} statement 
executes action $0$. Note that $S$ satisfies $R$. Set $R_{new}$ consists of a single
situation-action pair $r = (s,a)$ with $s$ such that $v(i_{|V|+1}) = True$, 
all remaining $v(i_j)$ are $False$, and $a = 2$. Finally, let $c_c = k + 1$.
Note that the instance of SRec-Spec$_D$ described above can be constructed in 
time polynomial in the size of the given instance of {\sc Dominating Set}. 

If there is a dominating set $D$ of size at most $k$ in the given instance of
{\sc Dominating Set}, construct a software system $S'$ consisting of a 
default selector and a single procedure modified from $S$ in which
the last \texttt{IF-THEN} statement executes action $2$ when $i_{|V|+1} = True$
and the conditions of the remaining $k$ \texttt{IF-THEN} statements are the 
situation-variables corresponding to the vertices in $D$ (with, if $|D| < k$, the
conditions of the final $k - |D|$ \texttt{IF-THEN} statements corresponding to
an arbitrary subset of size $k - |D|$ from $D$). Observe that $S'$ 
satisfies $R \cup R_{new}$, has $|sel| = 0$ and $|prc| \leq k + 1$, and can be
obtained by at most $k + 1$ changes to the system code of $S$.

Conversely, suppose that 
the constructed instance of SRec-Spec$_D$ has a software system $S'$ satisfying 
$R \cup R_{new}$ with $|sel| = 0$ and $|prc| \leq k + 1$ that can be created from
$S$ by at most $k + 1$ code changes. The selector in this system must be the 
default selector (as it is the only selector with $|sel| = 0$). Hence, all code
changes are confined to the procedure. This procedure has two possible
structures:

\begin{enumerate}
\item{
\textit{No negated situation-variable conditions}: As $r_{|V|+1}$ has
no situation-vari-ables set to $True$, it can only be processed correctly
by the final \texttt{ELSE}-statement, which must thus execute 
associated action $0$. As $r$ has no situation-variable set to $True$ except
$i_{|V|+1}$, it can only be processed correctly by an \texttt{IF-THEN}
statement with condition $i_{|V|+1}$ that executes action $2$. In order to 
correctly process all remaining situation-action pairs in $R$,
the remaining $\leq k$ {\tt IF-THEN} statements must both have associated executed action 1 and
situation-variables whose corresponding vertices form a dominating set in $G$
of size at most $k$.
}
\item{
\textit{Negated situation-variables are present}: 
Let $c$ be the first
negated situation-variable condition encountered moving down the
code in the procedure, $C$ be the set of unnegated situation-variable
conditions encountered before $c$, and $R' \subseteq R \cup \{r\}$ be the set of
situation-action pairs in $R \cup \{r\}$ that are not recognized by the 
situation-variables in $C$.
As $|prc| \leq k + 1$, $|C| \leq k$. Moreover, as
situation-action pair $r_{|V|+1}$ has no 
situation-variable with value $True$ and hence cannot be recognized
by an \texttt{IF-THEN} statement with an unnegated situation-variable
condition, $r_{|V|+1} \in R'$. 

Let $R'' \subset R \cup \{r\}$ be
such that $R' = R'' \cup \{r_{|V|+1}\}$. There are three possibilities:

\begin{enumerate}
\item{
\textit{$r \not\in R''$}: In this case, $i_{|V|+1}$
must be in $C$ in order for $r$ to be recognized prior to $c$. However, this
implies that all situation-action pairs in $R - \{r_{|V|+1}\}$ must have been
recognized prior to the \texttt{IF-THEN} statement with condition $i_{|V|+1}$
(otherwise, the action 2 required to correctly recognize $r$ would have caused
the wrong action to be output for at least one of these situation-action 
pairs). This would mean that the vertices corresponding to the
situation-variables in $C - \{i_{|V|+1}\}$ form a dominating set in $G$ of size at
most $k - 1$ in $G$. 
}
\item{
\textit{$r \in R''$ and $|R''| = 1$, i.e., $R'' = \{r\}$}: 
The variables in $C$ must have recognized all situation-action pairs in
$R$ except $r_{|V|+1}$, and hence the vertices corresponding to the
situation-variables in $C$ form a dominating set in $G$ of size at
most $k$ in $G$. 
}
\item{
\textit{$r \in R''$ and $|R''| > 1$}: Observe that the situation-variable in
$c$ cannot be in $I - \{i_{|V|+1}\}$
as the action executed by the 
associated \texttt{IF-THEN} statement would result in the procedure executing 
the wrong action for at least one situation-action pair in $R'$ (if $0$, at least
one of the situation-action pairs in $R''$; if $1$, $r_{|V|+1}$ and $r$; if $2$,
at least one of the situation-action pairs in $R' - \{r\}$). Hence, the only
situation-action variable that can safely occur negated in $c$ at this 
point is $i_{|V|+1}$, whose negation recognizes only $r_{|V|+1}$.

Consider now the processing of $R'' - \{r_{|V|+1}\}$ in the procedure after the
\texttt{IF-THEN} statement with condition $c$. As all remaining unrecognized
situation-action pairs have $i_{|V|+1} = True$ and cannot be recognized by
condition $\neg i_{|V|+1}$, we will assume that no
statement with condition $\neg i_{|V|+1}$ is encountered. There are three possible
cases: 

\begin{enumerate}
\item{
{\em An \texttt{IF-THEN} statement with condition $i_{|V|+1}$ is encountered}:
Let $C'$ be the set of unnegated situation-variable conditions encountered 
between $c$ and condition $i_{|V|+1}$ and $R''' \subseteq R'' - \{r_{|V|+1}\}$ be 
the set of situation-action pairs remaining to be processed as of condition 
$i_{|V|+1}$. As $i_{|V|+1} \not\in C'$, $r \in R'''$. Moreover, as $|prc| = k + 1$, 
$|C| + |C'| \leq k - 1$.

As $R'''$ does not contain $r_{|V|+1}$, $i_{|V|+1}$ is $True$ for all 
situation-action pairs in $R'''$. We also know that $S'$ satisfies $R$.
This implies that $R'''$ only has one member because
if this is not so, the action executed by the \texttt{IF-THEN} statement
with $i_{|V|+1}$ as its condition cannot be $1$ or $2$ 
without processing at least one member of $R'''$ incorrectly.
However, this implies that the vertices 
corresponding to the situation-variables in $C \cup C'$ form a 
dominating set in $G$ of size at most $k - 1$.
}
\item{
{\em Another \texttt{IF-THEN} statement with a negated situation-variable
condition $c'$ is encountered}: Let $x$ be the situation-variable negated in $c'$,
$C'$ be the set of unnegated situation-variable conditions encountered between
$c$ and $c'$, and $R''' \subseteq R'' - \{r_{|V|+1}\}$ be the set of situation-action
pairs remaining to be processed as of condition $c'$. We may assume that 
$i_{|V|+1} \not\in C'$ as this case is covered in (i) above; hence, $r \in R'''$. 
Moreover, as $|prc| \leq k + 1$, $|C| + |C'| \leq k - 1$.

The situation-variable $x$ cannot have value $False$ for any of the situation-action 
pairs in $R''' - \{r\}$ because having either $1$ or $2$ as the action executed by 
the associated \texttt{IF-THEN} statement would result in the procedure executing 
the wrong action for at least one situation-action pair in $R'''$ (if $1$, $r$; 
if $2$, at least one of the situation-action pairs in $R''' - \{r\}$). This 
implies that all situation-action pairs in $R''' - \{r\}$ have value $True$ for $x$,
which in turn implies that the vertices corresponding to the situation-variables in 
$C \cup C' \cup \{x\}$ form a dominating set in $G$ of size at most $k$.
}
\item{
{\em Neither (i) nor (ii) occurs before the final \texttt{ELSE} statement}:
Let $C'$ be the set of unnegated situation-variable conditions 
encountered between $c$ and the final \texttt{ELSE} statement and $R''' 
\subseteq R''$ be the set of situation-action pairs remaining to be processed 
as of the final \texttt{ELSE} statement.  As neither (i) nor (ii) occurred, $r$ 
cannot have been recognized before the final \texttt{ELSE} statement; hence, 
$r \in R'''$. Moreover, as $|prc| = k + 1$, $|C| + |C'| \leq k - 1$.

As $S'$ satisfies $R \cup \{r\}$, in order to process $r$ correctly, the final
\texttt{ELSE} statement must execute action $2$. However, this implies
that $r$ is the only member of $R'''$ because if this is not so, any
other member of $R'''$ (which would require the execution of action $1$)
would be processed incorrectly. However, this implies that the vertices 
corresponding to the situation-variables in $C \cup C'$ form a 
dominating set in $G$ of size at most $k$.
}
\end{enumerate}
}
\end{enumerate}
}
\end{enumerate}

\noindent
Hence, the existence of a satisfying software system for the
constructed instance of SRec-Spec$_D$ implies the existence of a
dominating set of size at most $k$ for the given instance  of
{\sc Dominating Set}.

To complete the proof, note that in the constructed instance of \linebreak 
SCre-Spec$_D$, $|A| = 3$, $|sel| = 0$, $d = 2$, $|R_{new}| = 1$, $|prc| = c_c = k + 1$,
and $|S| = (|sel| + 1)(|prc| + 2) = k + 3$.
\end{proof}

\vspace*{0.06in}

\begin{lemma}
{\sc Dominating Set} polynomial-time many-one reduces to \linebreak SRec-CompA$_D$ 
such that in the constructed instance of SRec-CompA$_D$, $|A| = 3$, $d = 2$, $|L_{sel}| = 
|L_{prc}| = |R_{new}| = 1$, $|sel| = c_l = 0$, and $|prc|$,  $|S|$, and $c_c$ are 
functions of $k$ in the given instance of {\sc Dominating Set}.
\label{LemRedSRecCompA} 
\end{lemma} 

\begin{proof}
Given an instance $\langle G = (V, E), k\rangle$ of {\sc Dominating Set}, 
the constructed instance $\langle I, A, R, S, L_{sel}, 
L_{prc}, d, R_{new}, c_c, c_l \rangle$ of SRec-CompA$_D$ has $I$, $A$, $R$, $S$, and 
$R_{new}$ as in the reduction in Lemma \ref{LemRedSRecSpec}, $L_{sel}$ and $L_{sel}$
consist only of the selector and procedure, respectively, in the given software 
system $S$ in the reduction in Lemma \ref{LemRedSRecSpec}, $d = 2$, $c_l = 0$, and 
$c_c = k + 1$. Note that the instance of SRec-Spec described above can be constructed
in time polynomial in the size of the given instance of {\sc Dominating Set}. As no 
component changes are possible, this constructed instance of SRec-CompA$_D$ is for
all intensive purposes an instance of SRec-Spec, and the proof of correctness of 
this reduction is identical to that given for the reduction in Lemma 
\ref{LemRedSRecSpec}. To complete the proof, note that in the constructed instance 
of SCre-CompA$_D$, $|A| = 3$, $|sel| = 0$, $d = 2$ $|R_{new}| = 1$,  $|L_{sel}| = 
|L_{prc}| = 0$, $|prc| = c_c = k + 1$, and $|S| = (|sel| + 1)(|prc| + 2) = k + 3$.
\end{proof}

\vspace*{0.06in}

\begin{lemma}
{\sc Dominating Set} polynomial-time many-one reduces to \linebreak SRec-Comp$_D$ such 
that in the constructed instance of SRec-Comp$_D$, $|A| = |L_{sel}| = 2$, $|prc| = 1$, 
and $d$ is a function of $k$ in the given instance of {\sc Dominating Set}.
\label{LemRedSRecComp} 
\end{lemma} 

\begin{proof}
Given an instance $\langle G = (V, E), k\rangle$ of {\sc Dominating Set}, the 
constructed instance $\langle I, A, R, S, R_{new}, L_{sel}, L_{prc}, c_l, d\rangle$ of 
SCre-Comp$_D$ has $I = \{i_1, i_2, \ldots$ \linebreak $i_{2|V| + 1}\}$ and 
$A = \{0, 1\}$. There are $|V|$ situation-action pairs in $R$ such that 
for $r_i = (s_i, a_i)$, $1 \leq i \leq |V|$, $v(i_i) = True$, $v(i_{|V| + j}) 
= True$ if $v_j \in N_C(v_i)$, $v(i_{2|V|+1}) = True$, all other 
$v(i_j)$ are $False$, and $a_i = 1$.
System $S$ consists of a selector whose $|V|$ \texttt{IF-THEN}
statement conditions are $i_{2|V|+1}$ followed by the first $|V| - 1$ 
situation-variables in $I$ and whose $|V| + 1$ associated procedures
(progressing along the selector) are all a procedure consisting
of an \texttt{IF-THEN} statement with condition $i_{2|V|+1}$ and action
$1$ followed by an \texttt{ELSE} statement with action 1.
Observe that $S$ satisfies $R$ (albeit trivially, as the condition
of the first \texttt{IF-THEN} statement in the selector always evaluates to 
$True$ and hence recognizes all situation-action pairs in $R$). Library $L_{sel}$ 
consists of the union of the selector in $S$ described above and the
selector whose $|V|$ \texttt{IF-THEN}
statement conditions are $\neg i_{2|V|+1}$ followed by the first $|V| - 1$ 
situation-variables in $I$. Library $L_{prc}$ consists of the union of the
procedure in $S$ described above and $L_{prc}$ as described in the reduction given
in Lemma \ref{LemRedSCreComp}. Let $R_{new}$ consist of $|V|$ situation-action
pairs in which $r_i = (s_i, a_i)$, $1 \leq i \leq |V|$, $v(i_i) = v(i_{2|V|+1})
= True$, all other $v(i_j)$ are $False$, and $a_i = 0$.  Finally, 
let $c_l = |V| + 1$ and $d = k + 1$. Note that the instance of SRec-Comp$_D$ described 
above can be constructed in time polynomial in the size of the given instance of 
{\sc Dominating Set}. 

If there is a dominating set $D$ of size at most $k$ in the given instance
of {\sc Dominating Set}, modify $S$ by swapping in the sole selector 
in $L_{sel}$ and assigning procedures to this selector such that each 
\texttt{IF-THEN} statement with condition $i_j$, $1 \leq j \leq |V| - 1$, executes 
one of the procedures in $L_{prc}$ corresponding to a vertex in $D$ that 
dominates $v_j$ and the \texttt{ELSE} statement executes one of the procedures in 
$L_{prc}$ corresponding to a vertex in $D$ that dominates $v_{|V|}$. Observe that 
this software system satisfies all situation-action pairs in $R \cup R_{new}$, has 
$d \leq k + 1$, and can be obtained from $S$ by $c_l$ component changes. 

Conversely, suppose that the constructed instance of SRec-Comp$_D$ has a software 
system $S'$ derived from $S$ by at most $c_l$ changes that has at most $d$ 
component-types and satisfies $R \cup R_{new}$. If there is a dominating set in 
$G$ consisting of a single vertex, system $S$ can be modified to accept $R \cup 
R_{new}$ by replacing the procedure executed by the first \texttt{IF-THEN} statement 
in the selector in $S$ with the procedure in $L_{prc}$ 
corresponding to that dominating vertex. If such a 
dominating vertex does not exist in $G$, as all members of $R \cup R_{new}$
have $v(i_{2|V|+1}) = True$, there is no procedure in $L_{prc}$ that can
replace the procedure executed by the first \texttt{IF-THEN} statement in the
selector in $S$ to process all situation-action pairs
in $R \cup R_{new}$ correctly. Hence, $S'$ must use the selector in $L_{sel}$.
As the condition of the first \texttt{IF-THEN} statement in this new selector
never evaluates to $True$ on any of the situation-action pairs in $R \cup
R_{new}$, all of these pairs must be recognized and processed correctly in
the remainder of the selector and its called procedures. As $d$ is at most
$k + 1$, this means that at most $k$ procedures must be taken from $L_{prc}$
to replace all procedures called following the first \texttt{IF-THEN}
statement in the selector in $S$; this entails a further $|V|$ component changes,
which brings the total required to change $S$ into $S'$ to $|V| + 1$. However, this 
all implies that the vertices encoded in these procedures form a dominating set 
in $G$ of size at most $k$.

To complete the proof, note that in the constructed instance of \linebreak 
SRec-Comp$_D$, $|A| = |L_{sel}| = 2$, $|prc| = 1$, and $d = k + 1$.
\end{proof}

\vspace*{0.06in}

\noindent
{\bf Result A}: If any of SCre-Spec, SCre-Comp, SCre-CompA, SRec-Spec, SRec-Comp, or 
                 SRec-CompA  is polynomial-time tractable then $P = NP$.

\vspace*{0.08in}

\begin{proof}
The $NP$-hardness of the decision versions of these problems follows from the 
$NP$-hardness of {\sc Dominating Set} and the reductions in Lemmas \ref{LemRedSCreSpec},
\ref{LemRedSCreComp}, \ref{LemRedSCreCompA}, \ref{LemRedSRecSpec}, \ref{LemRedSRecComp},
and \ref{LemRedSRecCompA}, respectively. The result then follows from Lemma
\ref{LemProp1}.
\end{proof}

\vspace*{0.1in}

\noindent
{\bf Result B}: If $P = BPP$ and either SCre-Spec, SCre-Comp, SCre-CompA, SRec-Spec, 
                 SRec-Comp, or SRec-CompA is polynomial-time tractable by a 
                 probabilistic algorithm which operates correctly with probability 
                 $\geq 2/3$ then $P = NP$.

\vspace*{0.08in}

\begin{proof}
It is widely believed that $P = BPP$ \cite[Section 5.2]{Wig07} where $BPP$ is 
considered the most inclusive class of decision problems that can be efficiently solved
using probabilistic methods (in particular, methods whose probability of correctness is
$\geq 2/3$ and can thus be efficiently boosted to be arbitrarily close to one). Hence, 
if any of SCre-Spec, SCre-Comp, SCre-CompA, SRec-Spec, SRec-Comp, or SRec-CompA 
has a probabilistic polynomial-time 
algorithm which operates correctly with probability $\geq 2/3$
then by the observation on which Lemma \ref{LemProp1} is based, their corresponding
decision versions also have such algorithms and are by definition in $BPP$. However,
if $BPP = P$ and we know that all these decision versions are $NP$-hard by the
proof of Result A, this would then imply by the definition of $NP$-hardness
that $P = NP$, completing the result.
\end{proof}

\vspace*{0.1in}

\noindent
{\bf Result C}: If $\langle |A|, |sel|, |prc|, |S|, d\rangle$-SCre-Spec is fp-tractable then
                 $FPT = W[1]$.

\vspace*{0.08in}

\begin{proof}
Follows from the $W[2]$-hardness of $\{ k \}$-{\sc Dominating Set}, the inclusion of 
$W[1]$ in $W[2]$, the reduction from {\sc Dominating Set} to SCre-Spec
given in Lemma \ref{LemRedSCreSpec}, and the definition of $FPT$.
\end{proof}

\vspace*{0.1in}

\noindent
{\bf Result D}: If $\langle |A|, |prc|, d, |L_{sel}|\rangle$-SCre-Comp is fp-tractable 
                 then $FPT = W[1]$.

\vspace*{0.08in}

\begin{proof}
Follows from the $W[2]$-hardness of $\{ k \}$-{\sc Dominating Set}, the inclusion of
$W[1]$ in $W[2]$, the reduction from {\sc Dominating Set} to SCre-CompA given in Lemma 
\ref{LemRedSCreComp}, and the definition of $FPT$.
\end{proof}

\vspace*{0.1in}

\noindent
{\bf Result E}: If $\langle |A|, |sel|, |prc|, |S|, d,$ $|L_{sel}|, |L_{prc}|, 
                 c_c \rangle$-SCre-CompA is fp-tractable then $FPT = W[1]$.

\vspace*{0.08in}

\begin{proof}
Follows from the $W[2]$-hardness of $\{ k \}$-{\sc Dominating Set}, the inclusion of
$W[1]$ in $W[2]$, the reduction from {\sc Dominating Set} to SCre-Comp
given in Lemma \ref{LemRedSCreCompA}, and the definition of $FPT$.
\end{proof}

\vspace*{0.1in}

\noindent
{\bf Result F}: If $\langle |A|, |sel|, |prc|, |S|, d, |R_{new}|, c_c\rangle$-SRec-Spec is 
                 fp-tractable then $FPT = W[1]$.

\vspace*{0.08in}

\begin{proof}
Follows from the $W[2]$-hardness of $\{ k \}$-{\sc Dominating Set}, the inclusion of
$W[1]$ in $W[2]$, the reduction from {\sc Dominating Set} to SRec-Spec given in Lemma 
\ref{LemRedSRecSpec}, and the definition of $FPT$.
\end{proof}

\vspace*{0.1in}

\noindent
{\bf Result G}: If $\langle |A|, |prc|, d, |L_{sel}|\rangle$-SRec-Comp is fp-tractable 
                 then $FPT = W[1]$.

\vspace*{0.08in}

\begin{proof}
Follows from the $W[2]$-hardness of $\{ k \}$-{\sc Dominating Set}, the inclusion of
$W[1]$ in $W[2]$, the reduction from {\sc Dominating Set} to SRec-CompA
given in Lemma \ref{LemRedSRecComp}, and the definition of $FPT$.
\end{proof}

\vspace*{0.1in}

\noindent
{\bf Result H}: If $\langle |A|, |sel|, |prc|, |S|, d, |L_{sel}, |L_{prc}|, |R_{new}|, c_l, 
                 c_c \rangle$-SRec-CompA is \newline fp-tractable then $FPT = W[1]$.

\vspace*{0.08in}

\begin{proof}
Follows from the $W[2]$-hardness of $\{ k \}$-{\sc Dominating Set}, the inclusion of
$W[1]$ in $W[2]$, the reduction from {\sc Dominating Set} to SRec-Comp given in Lemma 
\ref{LemRedSRecCompA}, and the definition of $FPT$.
\end{proof}

\vspace*{0.1in}

\noindent
{\bf Result I}: $\langle I\rangle$-SCre-Spec, -SCre-Comp, -SCre-CompA, -SRec-Spec, 
                    -SRec-Comp, and -SRec-CompA are fp-tractable.

\vspace*{0.08in}

\begin{proof}
Let $S(|I|)$ be the set of software systems with the structure considered in this paper
that can be constructed using $|I|$ situation-variables and $|S(|I|)|$ be the number of
such systems. The maximum number of conditions that can be evaluated for any 
situation-action pair is $|I| + 1$ as once any condition based on $x \in I$ occurs, at 
most $|I| - 1$ additional conditions not involving $x$ can occur before condition 
$\neg x$ must occur. Once $\neg x$ occurs, as each $r \in R$ satisfies either $x$ or 
$\neg x$, further conditions are not evaluated. Given that these $|I| + 1$ conditions 
that can occur in any selector or procedure are drawn from $2|I|$ candidates (each $x 
\in I$ and its negation), there are $\left( \begin{array}{c} 2|I| \\ |I| + 1
\end{array} \right) \leq (|2|I|)^{|I| + 1}$ possible selections of
conditions and $(|I| + 1)! \leq (|I| + 1)^{|I| + 1}$ possible orderings
of these conditions. There are thus less than $(2|I|)^{|I| + 1}(|I| + 1)^{|I|+1}$
selectors or procedures with exactly $|I| + 1$ conditions, and less than 
$(|I| + 1)(2|I|)^{|I| + 1}(|I| + 1)^{|I|+1}$ selectors or procedures
with at most $|I| + 1$ conditions. Let us denote the latter quantity with $T$.
Given a selector, there are less than $T^{|I| + 2}(|I|+2)^{|I| + 2}$ 
systems that can be built with that selector, and as there are less than $T$
selectors, $|S(|I|)| \leq T^{|I| + 3}(|I|+2)^{|I|+2}$. Let us denote this
(extraordinarily loose) upper bound by $T'$.

A basic algorithm for each of the problems examined in this paper is to
consider each system $S'$ in $S(|I|)$, determine if it satisfies $R$, and
then check if it satisfies any additional required properties for a solution, i.e., 

\begin{enumerate}
\item at most $d$ component-types occur in $S'$ (SCre-CompA, SCre-Comp,
       SRec-CompA, SRec-Comp);
\item at most $c_l$ component changes were made in transforming $S$ into
       $S'$ (SRec-CompA, SRec-Comp); and
\item at most $c_c$ code changes were made in transforming $S$ into $S'$
       (SCre-CompA, SRec-Spec, SRec-CompA).
\end{enumerate}

\noindent
Test (1) can be done in time polynomial in the problem input size. Test (2) can
be performed for each $(S,S')$ pair by generating every possible
non-repeating chain of intermediate systems linking $S$ and $S'$ of
length at most $c_l$ and checking that each adjacent pair of systems in this
chain can be generated by a one-component change relative to given 
$L_{sel}$ and $L_{prc}$ (both of whose sizes are upper-bounded by $T'$). 
As the longest possible chain of systems has each possible system occurring once,
$c_l \leq T'$. The number of such
chains is at most $c_l(T'^{c_l})c_l^{c_l} \leq T'^{T'^2 + 1}$ and the
one-component check can be done in time polynomial in $|I|$. By a similar
argument we can establish that test (3) can be done in a similar amount
of time. As both $T$ and $T'$ are functions of $|I|$, the running times
of the basic algorithm sketched above is upper bounded by a function of
$|I|$ times some polynomial of the input size for each of the problems
considered here. Hence, each of the problems considered here when
parameterized relative to $|I|$ is in $FPT$.
\end{proof}

\vspace*{0.1in}

\noindent
{\bf Result J}: $\langle |sel|, |L_{prc}|\rangle$-SCre-Comp and -SRec-Comp are 
                    fp-tractable.

\vspace*{0.08in}

\begin{proof}
Let $S(L_{sel}, |sel|, L_{prc})$ be the set of two-level software systems 
that can be constructed from a set $L_{sel}$ of selectors with at most $|sel|$ 
conditions relative to a procedure-library $L_{prc}$ and $|S(L_{sel}, |sel|, L_{prc})|$
be the number of such systems. For any individual selector with at most $|sel|$
conditions, there will be at most $|L_{prc}|^{|sel|}$ possible ordered selections of 
procedures from $L_{prc}$.  Given that there are $|L_{sel}|$ selectors to choose, 
$|S(L_{sel}, |sel|, L_{prc})| \leq |L_{sel}||L_{prc}|^{|sel|}$.
Given this, by an argument slightly modified from that given in the proof of 
Result I, we can establish that the basic algorithms in the proof of Result I 
implementing tests (1) and/or (2) as appropriate for problems SCre-Comp and SRec-Comp
run in time that is upper-bounded by a function of $|sel|$ and $|L_{prc}|$ times some
polynomial of the problem input size. Hence, both of these problems when
parameterized relative to $|sel|$ and $|L_{prc}|$ are in $FPT$.
\end{proof}

\end{document}